\documentclass[11pt]{article}
\usepackage{graphics,epsfig}
\usepackage[]{graphicx}
\usepackage{latexsym}
\usepackage{amsmath, amsthm, amsfonts, amssymb}
\usepackage{verbatim}

\begin{document}

\newtheorem{theo}{Theorem}[section]
\newtheorem{definition}[theo]{Definition}
\newtheorem{lem}[theo]{Lemma}
\newtheorem{prop}[theo]{Proposition}
\newtheorem{coro}[theo]{Corollary}
\newtheorem{exam}[theo]{Example}
\newtheorem{rema}[theo]{Remark}
\newtheorem{example}[theo]{Example}
\newtheorem{principle}[theo]{Principle}
\newcommand{\ninv}{\mathord{\sim}} %involutive negation
\newtheorem{axiom}[theo]{Axiom}
\numberwithin{equation}{subsection}
\title{Convex Quantum Logic}

\author{{\sc Federico Holik}$^{1}$ \ {\sc ,} \ {\sc Cesar Massri}$^{2}$ \ {\sc
and} \ {\sc Nicol\'{a}s Ciancaglini}$^{1}$}

\maketitle

\begin{center}

\begin{small}
1- Instituto de Astronom\'{\i}a y F\'{\i}sica del Espacio (IAFE)\\
Casilla de Correo 67, Sucursal 28, 1428 Buenos Aires, Argentina\\
2- Departamento de Matem\'{a}tica - Facultad de Ciencias Exactas y
Naturales\\ Universidad de Buenos Aires - Pabell\'{o}n I, Ciudad
Universitaria \\ Buenos Aires, Argentina
\end{small}
\end{center}

\vspace{1cm}

\begin{abstract}
\noindent In this work we study the convex set of quantum states
from a quantum logical point of view. We consider an algebraic
structure based on the convex subsets of this set. The relationship
of this algebraic structure with the lattice of propositions of
quantum logic is shown. This new structure is suitable for the study
of compound systems and shows new differences between quantum and
classical mechanics. This differences are linked to the nontrivial
correlations which appear when quantum systems interact. They are
reflected in the new propositional structure, and do not have a
classical analogue. This approach is also suitable for an algebraic
characterization of entanglement.
\end{abstract}
\bigskip
\noindent

\begin{small}
\centerline{\em Key words: entanglement-quantum logic-convex sets}
\end{small}

\bibliography{pom}

\section{Introduction}

There exists a widely discussed distinction between proper and
improper mixtures \cite{d'esp, Mittelstaedt}. In this work we adopt
this distinction as a starting point. From this point of view,
improper mixtures do not admit an ignorance interpretation, as is
the case for proper mixtures (for more discussion on this subject
see also \cite{Kirkpatrik}, \cite{ReplytoKirkpatrik} and
\cite{tanosquaternionicos}). One of the most important consequences
of this fact is that improper mixtures have to be considered as
states on their own right, besides pure states.

Furthermore, the study of the convex set of quantum states plays a
central role in decoherence \cite{Garchauer} and quantum information
processing \cite{horodecki2001}, and this convex set is formed
mainly by mixed states. For a Hilbert space of finite dimension $N$,
pure states form a $2(N-1)$-dimensional subset of the
$(N^{2}-2)$-dimensional boundary of the convex set of states. Thus,
\emph{regions} of the convex set of states are the important objects
to investigate, rather than the lattice of projections.

The standard quantum logical approach to quantum mechanics ($QM$)
\cite{BvN},\cite{dallachiaragiuntinilibro},\cite{jauch},
\cite{piron}, takes the lattice of projections of the Hilbert space
of the system to be the lattice of propositions (see section
\ref{s:introduction to QL} for a brief review). In this framework,
as in classical mechanics ($CM$), the state of the system is in
direct correspondence with the conjunction of all its actual
properties. This conjunction yields a pure state or a ray in the
Hilbert space, which corresponds to an atom of the lattice of
projections. But when we consider the system $S$ formed by
subsystems $S_{1}$ and $S_{2}$, the following problem appears (see
for example \cite{aertsparadox}). If $S$ is in an entangled (pure)
state, the sates of the subsystems become improper mixtures, and so,
it is no longer true that the conjunction of all actual properties
yields the true state of the subsystem. On the contrary, the
conjunction yields a non atomic proposition of the von Newmann
lattice of projections ($\mathcal{L}_{v\mathcal{N}}$), i.e., a
subspace of dimension greater than one \cite{aertsjmp84}, and leaves
the state undetermined. Thus, this procedure does not give the real
physical state of the system any more. We analyze this problem in
detail in section \ref{s:Explicamos por que}, where we review and
extend the discussion posed in \cite{extendedql}, and impose general
conditions for the structures that we want to construct.

The usual way to incorporate improper mixtures (or more generally,
states) in the $QL$ approach, is as measures over
$\mathcal{L}_{v\mathcal{N}}$ (or equivalently, measures over the
propositions of the abstract lattice). So, in the study of cases
which may involve improper mixtures, we have to jump to a different
level than that of the propositions (projections). The physical
propositions belong to $\mathcal{L}_{v\mathcal{N}}$, while mixtures
belong to the set of measures over $\mathcal{L}_{v\mathcal{N}}$. We
claim that this split is at the heart of the problem posed in
\cite{aertsparadox}, and complementarily, in section
\ref{s:Explicamos por que} of this article. We will review this in
detail in sections \ref{s:introduction to QL} and \ref{s:Explicamos
por que}.

In section \ref{s:Explicamos por que} we give a list of conditions
for the structures that we are looking for, in order to grant that
they solve the posed problems. We work out a structure which
incorporates improper mixtures in the same status as pure states
(alike $\mathcal{L}_{v\mathcal{N}}$), i.e., in which all states of
the system are atoms of the lattice. A first approach in this
direction was done in \cite{extendedql}, where an extension of the
von Newmann lattice of projections was built. In that article it was
shown that it is possible to construct a lattice theoretical
framework which incorporates improper mixtures as atoms, denoted by
$\mathcal{L}$ here. We briefly reproduce (without proof), some of
the results of \cite{extendedql} in section \ref{s:New language}.

In this work we extend the construction of \cite{extendedql} to a
larger structure $\mathcal{L}_{\mathcal{C}}$, formed by all convex
subsets of $\mathcal{C}$. This is developed in sections
\ref{s:Convex lattice} and \ref{s:The Relationship for convex}. As
desired, this extension solves the problem possed in section
\ref{s:Explicamos por que}. We think that this approach is suitable
in order to consider decoherence or entangled systems from a quantum
logical and algebraic point of view. Furthermore, taking the convex
set of states as an starting point could be of interest if we take
into account that there exists a formulation of $QM$ in terms of
convex sets (see \cite{MielnikGQS}, \cite{MielnikTF} and
\cite{MielnikGQM}). This is an independent formulation of $QM$ and
has the advantage that it can include models of theories which
cannot be represented by Hilbert spaces, as is the case of non
linear generalizations of quantum mechanics.

Using $\mathcal{L}_{\mathcal{C}}$, we can construct projection
functions from the lattice of the whole system to the lattices of
the subsystems which satisfy, in turn, to be compatible with the
physical description. A similar construction can be made for
$\mathcal{L}$ (see section \ref{s:New language}). Alike the von
Newmann case, where these projection functions do not exist, the
projections defined in $\mathcal{L}$ and $\mathcal{L}_{\mathcal{C}}$
satisfy this condition. They are also canonical in the sense that
they are constructed using partial traces, in accordance with the
quantum formalism.

The approach presented here shows (as well as the one presented in
\cite{extendedql}), the radical difference between quantum mechanics
and classical mechanics when two systems interact, a difference
which is not properly expressed in the orthodox $QL$ approach. When
dealing with classical systems, no enlargement of the lattice of
propositions is needed even in the presence of interactions. The
phase space is sufficient in order to describe all relevant physics
about the subsystems. But the existence of non-trivial correlations
in quantum mechanics forces an enlargement of the state space of
pure states to the convex set $\mathcal{C}$, and so, the enlargement
of $\mathcal{L}_{v\mathcal{N}}$ to a structure as $\mathcal{L}$ or
$\mathcal{L}_{\mathcal{C}}$ (or any other structure which satisfies
conditions listed in section \ref{s:Explicamos por que}).

In section \ref{s:entanglement} we study the maps between
$\mathcal{L}_{\mathcal{C}}$, $\mathcal{L}_{\mathcal{C}_{1}}$ and
$\mathcal{L}_{\mathcal{C}_{2}}$ (the lattices of $S$ and its
subsystems), and show that our construction allows for an algebraic
characterization of entanglement, showing a new feature which is not
so explicit in the standard $QL$ approach. Finally, in section
\ref{s:Conclusions} we expose our conclusions.

\section{The Convex Set of States and Improper Mixtures}\label{s:introduction to QL}

Let us review first the quantum logical approach to the description
of physical systems (see for example
\cite{dallachiaragiuntinilibro}). In the standard $QL$ approach,
properties (or propositions) of a quantum system are in
correspondence with closed subspaces of a Hilbert space
$\mathcal{H}$. The set of subspaces ${\mathcal{P}}({\mathcal{H}})$
with the partial order defined by set inclusion $\subseteq$,
intersection of subspaces $\cap$ as the lattice meet, closed linear
spam of subspaces $\oplus$ as the lattice join and
orthocomplementation $^{\bot}$ as lattice complement, gives rise to
an orthomodular lattice
${\mathcal{L}}_{v\mathcal{N}}({\mathcal{H}})=
<{\mathcal{P}}({\mathcal{H}}),\ \cap,\ \oplus,\ \neg,\ 0,\ 1>$ where
$0$ is the empty set $\emptyset$ and $1$ is the total space
$\mathcal{H}$. This is the Hilbert lattice, named $QL$ by Birkhoff
and von Neumann. We will refer to this lattice as
${\mathcal{L}}_{vN}$, the `von Neumann lattice' (or simply
$\mathcal{L}(\mathcal{H})$).

Mixed states represented by density operators had a secondary role
in the classical treatise by von Newmann because they did not add
new conceptual features to pure states. In fact, in his book,
mixtures meant ``statistical mixtures'' of pure states {\rm
\cite[pg. 328]{vN}}, which are known in the literature as ``proper
mixtures'' {\rm \cite[Ch. 6]{d'esp}}. They usually represent the
states of realistic physical systems whose preparation is not well
described by pure states. In the standard formulation of $QL$,
mixtures (as well as pure states) are included as measures over the
lattice of projections\cite{mikloredeilibro}, that is, a state $s$
is a function:

\begin{equation}
s:\mathcal{L}(\mathcal{H})\longrightarrow [0;1]
\end{equation}

\noindent such that:

\begin{enumerate}
\item $s(\textbf{0})=0$ ($\textbf{0}$ is the null subspace).

\item For any pairwise orthogonal family of projections ${P_{j}}$,
$s(\sum_{j}P_{j})=\sum_{j}s(P_{j})$
\end{enumerate}

But while pure states can be put in a bijective correspondence to
the atoms of $\mathcal{L}(\mathcal{H})$, this is not the case for
mixtures. We review in Section \ref{s:Explicamos por que} how this
difference leads to problems when compound systems are considered.
We must pay attention to improper mixtures (\cite{d'esp},
\cite{Mittelstaedt}) because we have to deal with them in each (non
trivial) case in which a part of the system is considered.

For a classical system with phase space $\Gamma$, the lattice of
propositions is defined as the the set of subsets of $\Gamma$
($\mathcal{P}(\Gamma)$), endowed with set intersection as
conjunction ``$\wedge$", set union as disjunction ``$\vee$" and set
complement as negation ``$\neg$". We will call this lattice
$\mathcal{L}_{\mathcal{C}\mathcal{M}}$. The points $(p,q)\in\Gamma$
are in a bijective correspondence with the states of the system.
Statistical mixtures are represented as measurable functions:

\begin{equation}\label{classical statisticalmixture}
\sigma:\Gamma\longrightarrow [0;1]
\end{equation}

\noindent such that $\int_{\Gamma}\sigma(p,q)d^{3}pd^{3}q=1$.

For quantum compound systems $S_{1}$ and $S_{2}$, given the Hilbert
state spaces $\mathcal{H}_{1}$ and $\mathcal{H}_{2}$ as
representatives of two systems, the pure states of the compound
system are given by rays in the tensor product space
$\mathcal{H}=\mathcal{H}_{1}\otimes\mathcal{H}_{2}$. It is not true
that any pure state of the compound system factorizes after the
interaction in pure states of the subsystems. This situation is very
different from that of classical mechanics, where for state spaces
$\Gamma_{1}$ and $\Gamma_{2}$, we assign
$\Gamma=\Gamma_{1}\times\Gamma_{2}$ for the compound system.

Let us now briefly review the relationship between the states of the
joint system and the states of the subsystems in the quantum
mechanical case. Let us focus for simplicity on the case of two
systems, $S_{1}$ and $S_{2}$. If $\{|x_{i}^{(1)}\rangle\}$ and
$\{|x_{i}^{(2)}\rangle\}$ are the corresponding orthonormal basis of
$\mathcal{H}_{1}$ and $\mathcal{H}_{1}$ respectively, then the set
$\{|x_{i}^{(1)}\rangle\otimes|x_{j}^{(2)}\rangle\}$ forms an
orthonormal basis for $\mathcal{H}_{1}\otimes\mathcal{H}_{2}$. A
general (pure) state of the composite system can be written as:

\begin{equation}
\rho=|\psi\rangle\langle\psi|
\end{equation}

\noindent where
$|\psi\rangle=\sum_{i,j}\alpha_{ij}|x_{i}^{(1)}\rangle\otimes|x_{j}^{(2)}\rangle$.
And if $M$ represents an observable, its mean value $\langle
M\rangle$ is given by:

\begin{equation}\label{e:meanvalueoperator}
\mbox{tr}(\rho M)=\langle M\rangle
\end{equation}

If observables of the form $O_{1}\otimes\mathbf{1}_{2}$ and
$\mathbf{1}_{1}\otimes O_{2}$ (with $\mathbf{1}_{1}$ and
$\mathbf{1}_{2}$ the identity operators over $\mathcal{H}_{1}$ and
$\mathcal{H}_{2}$ respectively) are considered, then partial state
operators $\rho_{1}$ and $\rho_{2}$ can be defined for systems
$S_{1}$ and $S_{2}$. The relation between $\rho$, $\rho_{1}$ and
$\rho_{2}$ is given by:

\begin{equation}
\rho_{1}=tr_{2}(\rho) \ \ \ \ \rho_{2}=tr_{1}(\rho)
\end{equation}

\noindent where $tr_{i}$ stands for the partial trace over the $i$
degrees of freedom.  It can be shown that:

\begin{equation}
tr_{1}(\rho_{1}O_{1}\otimes\mathbf{1}_{2})=\langle O_{1}\rangle
\end{equation}

\noindent and that a similar equation holds for $S_{2}$. Operators
of the form $O_{1}\otimes\mathbf{1}_{2}$ and $\mathbf{1}_{1}\otimes
O_{2}$ represent magnitudes related to $S_{1}$ and $S_{2}$
respectively. When $S$ is in a product state
$|\varphi_{1}\rangle\otimes|\varphi_{2}\rangle$, the mean value of
the product operator $O_{1}\otimes O_{2}$ will yield:

\begin{equation}
\mbox{tr}(|\varphi_{1}\rangle\otimes|\varphi_{2}\rangle\langle\varphi_{1}
|\otimes\langle\varphi_{2}|O_{1}\otimes O_{2})=\langle
O_{1}\rangle\langle O_{2}\rangle
\end{equation}

\noindent reproducing statistical independence.

Mixtures are represented by positive, Hermitian and trace one
operators, (also called `density matrices'). The set of all density
matrixes forms a convex set (of states), which we will denote by
$\mathcal{C}$. Remember that the physical observables are
represented by elements of $\mathcal{A}$, the $\mathbb{R}$-vector
space of Hermitian operators acting on $\mathcal{H}$:

\begin{definition}\label{d:hermitian}
$\mathcal{A}:=\{ A\in B(\mathcal{H})\,|\, A=A^{\dagger} \}$
\end{definition}

\begin{definition}\label{d:mathcalC}
$\mathcal{C}:=\{\rho\in\mathcal{A}\,|\,\mbox{tr}(\rho)=1,\,\rho\geq
0\}$
\end{definition}

\noindent where $B(\mathcal{H})$ stands for the algebra of bounded
operators in $\mathcal{H}$. The set of pure states satisfies:

\begin{equation}
P:=\{\rho\in\mathcal{C}\,|\, \rho^{2}=\rho\}
\end{equation}

This set is in correspondence with the rays of $\mathcal{H}$ by the
association:

\begin{equation}
\mathcal{F}:\mathbb{C}\mathbb{P}(\mathcal{H})\longmapsto \mathcal{C}
\quad|\quad [|\psi\rangle]\longmapsto|\psi\rangle\langle\psi|
\end{equation}

\noindent where $\mathbb{C}\mathbb{P}(\mathcal{H})$ is the
projective space of $\mathcal{H}$, and $[|\psi\rangle]$ is the class
defined by the vector $|\psi\rangle$
($|\varphi\rangle\sim|\psi\rangle\longleftrightarrow|\varphi\rangle=\lambda|\psi\rangle$,
$\lambda\neq 0$). $\mathcal{C}$ is a convex set inside the
hyperplane $\{\rho\in\mathcal{A}\,|\,\mbox{tr}(\rho)=1\}$ formed by
the intersection of this hyperplane with the cone of positive
matrixes.

Separable states are defined \cite{Werner} as those states of
$\mathcal{C}$ which can be written as a convex combination of
product states:

\begin{equation}
\rho_{Sep}=\sum_{i,j}\lambda_{ij}\rho_{i}^{(1)}\otimes\rho_{j}^{(2)}
\end{equation}

\noindent where $\rho_{i}^{(1)}\in\mathcal{C}_{1}$ and
$\rho_{j}^{(2)}\in\mathcal{C}_{2}$, $\sum_{i,j}\lambda_{ij}=1$ and
$\lambda_{ij}\geq 0$. So, the set $\mathcal{S}(\mathcal{H})$ of
separable states is defined as:

\begin{definition}
$\mathcal{S}(\mathcal{H}):=\{\rho\in\mathcal{C}\,|\,\rho\,\,
\mbox{is separable}\}$
\end{definition}

As said above, it is a remarkable fact that there are many states in
$\mathcal{C}$ which are non separable. If the state is
non-separable, it is said to be entangled
\cite{bengtssonyczkowski2006}. The estimation of the volume of
$\mathcal{S}(\mathcal{H})$ is of great interest (see --among
others--\cite{zyczkowski1998}, \cite{horodecki2001} and
\cite{aubrum2006}).

In classical mechanics  mixtures are not of a fundamental nature.
They represent an state of ignorance of the observer, because we
know in principle that the system is in a given state $s$ of phase
space. On he other hand, in quantum mechanics, we must take into
account the difference between proper and improper mixtures. A
proper mixture can be considered as a density matrix \emph{plus} a
piece of classical information, which encodes classical
probabilities for preparations of ensembles of pure states. This
extra piece of classical information may have its source on
imperfections of the preparation procedure, or could be produced
deliberately, but the important fact is that these probabilities
could be determined -at least- in principle, as is the case of
mixtures in classical mechanics. But for the case of improper
mixtures, this information does not exist in the world.

What is one of the main implications of the fact that improper
mixtures do not admit an ignorance interpretation? For the standard
formulation of $QM$ we have at hand what it is usually called ``the
superposition principle":

\begin{principle}\label{e:superposition principle}
Superposition Principle. If $|\psi_{1}\rangle$ and
$|\psi_{1}\rangle$ are physical states, then
$\alpha|\psi_{1}\rangle+\beta|\psi_{1}\rangle$
($|\alpha|^{2}+|\beta|^{2}=1$) will be a physical state too.
\end{principle}

Are there other operations which allows us to form new states up
from two given states? If we accept that improper mixtures are
states of a fundamental nature as much as pure states do, then, the
fact that we can create new physical states \emph{mixing} two given
states, could be thought as a principle which stands besides the
superposition principle:

\begin{principle}\label{e:mixing principle}
Mixing Principle. If $\rho$ and $\rho'$ are physical states, then
$\alpha\rho+\beta\rho'$ ($\alpha+\beta=1$, $\alpha,\beta\geq 0$)
will be a physical state too.
\end{principle}

Mixing principle is not contained directly in the superposition
principle. Mixing principle appears as a consequence of the axiom
which states that to a compound system corresponds the tensor
product of Hilbert spaces. It expresses the fact that improper
mixtures are physical states. We will not consider proper mixtures
in this work, we only concentrate in physical states.

There is a remarkable physical consequence of all this (which we
think is not properly emphasized in the literature). While for pure
states there always exist ``true propositions" \cite{piron}, i.e.,
propositions for which a test will yield the answer ``yes" with
certainty (and a similar situation for ``false propositions"), the
situation is radically different for improper mixtures. If we accept
that improper mixtures are states of a fundamental nature as well as
pure states, then we must face the fact that there exist states for
which no ``true propositions" exist (discarding the trivial
proposition represented by the Hilbert space itself). This is the
case for example for the maximum uncertainty state (finite
dimension), $\rho=\frac{1}{N}\mathbf{1}$.
%Let us review with an example how the mixing principle operates. If
%$\rho=|\psi_{1}\rangle\langle\psi_{1}|$,
%$\rho'=|\psi_{2}\rangle\langle\psi_{2}|$ and
%$=|\psi\rangle\langle\psi|$ are pure states, consider the states

%\begin{equation}\label{e:rhosup}
%\rho_{sup}=(\alpha|\psi_{1}\rangle+\beta|\psi_{1}\rangle)(\alpha^{\ast}\langle\psi_{1}|+\beta^{\ast}\langle\psi_{1}|\,\,|\alpha|^{2}+|\beta|^{2}=1
%\end{equation}

%\begin{equation}\label{e:rhomix}
%\rho_{mix}=\alpha\rho+\beta\rho'\,\,\alpha+\beta=1,\,\,
%\alpha,\beta\geq 0
%\end{equation}

%\noindent Let us compute now the probabilities $P()$ (otro dia lo
%escribo).

\begin{figure}\label{f:going down on classical systems}
\begin{center}
\unitlength=1mm
\begin{picture}(5,5)(0,0)
\put(-15,23){\vector(-1,-1){20}} \put(15,23){\vector(1,-1){20}}
\put(0,25){\makebox(0,0){$\mathcal{L}_{\mathcal{CM}}=\mathcal{L}_{\mathcal{CM}_{1}}\times\mathcal{L}_{\mathcal{CM}_{2}}$}}
\put(-37,0){\makebox(0,0){$\mathcal{L}_{\mathcal{CM}_{1}}$}}
\put(37,0){\makebox(0,0){$\mathcal{L}_{\mathcal{CM}_{2}}$}}
\put(-25,16){\makebox(0,0){$\pi_{1}$}}
\put(25,16){\makebox(0,0){$\pi_{2}$}}
\end{picture}
\caption{In the classical case, we can map a state of the system to
the states of the subsystems using the set-theoretical projections
$\pi_{1}$ and $\pi_{2}$}.
\end{center}
\end{figure}
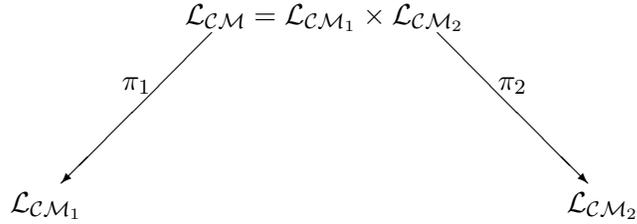

\section{The Limits of $\mathcal{L}_{v\mathcal{N}}$}\label{s:Explicamos por que}

In the standard $QL$ approach there is a bijective correspondence
between the atoms of the lattice ${\mathcal{L}}_{v\mathcal{N}}$ and
pure states. Each pure state $s$, which is represented by a ray
$[|\psi\rangle]$, is an atom of the lattice of projections
${\mathcal{L}}_{v\mathcal{N}}$. The relationship between the atoms
of the lattice and actual properties $p$ of the system (represented
by closed subspaces of $\mathcal{H}$) is given by:

\begin{equation}\label{e:ydelaspropiedades}
\{s\}=\wedge\{p\in \mathcal{L}_{v\mathcal{N}}\,|\,
p\,\,\mbox{is}\,\,\mbox{actual}\}
\end{equation}

\noindent and a similar relation holds in the classical case. But
there appears a subtle problem with equation
\ref{e:ydelaspropiedades} when compound systems are considered.
Suppose that $S_{1}$ and $S_{2}$ are subsystems of a larger system
$S$ in a pure entangled state $|\psi\rangle$. What happens if we
want to determine the states of its subsystems using equation
\ref{e:ydelaspropiedades}? This problem is studied in
\cite{aertsparadox}. If we make the conjunction of all actual
properties for, say $S_{1}$, we will no longer obtain an atom of
$\mathcal{L}_{v\mathcal{N}_{1}}$. Instead of it, we will obtain a
property which corresponds to a subspace of dimension greater than
one (see theorem $18$ of \cite{aertsjmp84}). But this property does
no longer corresponds to the actual state of the subsystem. The
state of the subsystem is not a pure state, but an improper mixture
given by the partial trace
$\mbox{tr}_{2}(|\psi\rangle\langle\psi|)$. So, it is not possible to
obtain the actual physical state of $S_{1}$ using the properties of
$\mathcal{L}_{v\mathcal{N}_{1}}$ and equation
\ref{e:ydelaspropiedades}.

Consider figures \ref{f:going down on classical systems} and
\ref{f:traces do not work}. Alike the classical case, in general, we
will not be able to map states from $\mathcal{L}_{v\mathcal{N}}$
into states of $\mathcal{L}_{v\mathcal{N}i}$ ($i=1,2$) using partial
traces (which are the physical maps that we \emph{should} use). This
is because the states of the subsystems are represented by improper
mixtures which are not projections, and thus, they do not belong to
$\mathcal{L}_{v\mathcal{N}i}$ ($i=1,2$). How can we complete the
``?" symbols of figure \ref{f:traces do not work}? There is no way
to do that at the lattices level (when we use von Newmann lattices).
We should have to jump into the level of measures over
$\mathcal{L}_{v\mathcal{N}_{1}}$. But in this work we want to avoid
this possibility.

As said above, statistical mixtures of $CM$ do not have a
\emph{fundamental nature}; on the contrary, they represent a state
of knowledge of the observer, a loss of information. Alike classical
mixtures, improper mixtures are of a fundamental nature; according
with the orthodox interpretation of $QM$, there is no other
information available about the state of the subsystem. As well as
pure states (and points of classical phase space), they represent
pieces of information that are maximal and logically complete. They
cannot be consistently extended and they decide any property or the
result of any experiment on the given subsystem. They
\emph{determine the physics of the subsystem}. Due to this, we want
to consider improper mixtures as states in the same level to that of
pure states, i.e., we want that they belong to the same
propositional structure.

The strategy that we follow in this work is to search for structures
which contain improper mixtures in such a way that they have an
equal treatment as the one given to pure states. As we will see,
this is possible, and such structures can be defined in a natural
way, extending (in a sense explained below)
$\mathcal{L}_{v\mathcal{N}}$ and in a way which is compatible with
the physics of compound quantum systems.

We want to avoid the fact that actual properties of the
propositional system do not determine the state of the system,
understood as the state of affairs which determines its physics. We
think that every reasonable notion of \emph{physical} state in a
propositional system should satisfy equation
\ref{e:ydelaspropiedades}.

As remarked above, there are improper mixtures for which all yes-no
tests are uncertain. But it is important to remark that this does
not imply that the system has no testable properties at all. Making
quantum state tomographies we can determine the state of the system.
These kind of ``tests" however, are of a very different nature than
that of the yes-no experiments. But the only thing that we care
about is that of the reality of physical process and our capability
of experimentally test this reality. We search for structures which
reflect this physics in a direct way.

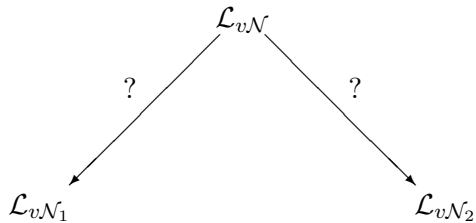
\begin{figure}\label{f:traces do not work}
\begin{center}
\unitlength=1mm
\begin{picture}(5,5)(0,0)
\put(-3,23){\vector(-1,-1){20}} \put(3,23){\vector(1,-1){20}}
\put(0,25){\makebox(0,0){$\mathcal{L}_{v\mathcal{N}}$}}
\put(-27,0){\makebox(0,0){${\mathcal{L}_{v\mathcal{N}_{1}}}$}}
\put(27,0){\makebox(0,0){${\mathcal{L}_{v\mathcal{N}_{2}}}$}}
\put(-15,16){\makebox(0,0){$?$}} \put(15,16){\makebox(0,0){$?$}}
\end{picture}
\caption{We cannot apply partial traces in order to go down from
$\mathcal{L}_{v\mathcal{N}}$ to $\mathcal{L}_{v\mathcal{N}_{1}}$,and
$\mathcal{L}_{v\mathcal{N}_{2}}$. How do we complete the ``?"
symbols?}
\end{center}
\end{figure}

There are other reasons for considering structures which contains
improper mixtures in a same status as that of pure states. There are
a lot of studies of interest which concentrate on mixtures. For
example, this is the case in quantum decoherence, quantum
information processing, or the independent generalizations of
quantum mechanics which emphasize the convex nature of mechanics
(not necessarily equivalent to ``Hilbertian" $QM$). The set of
interest in these studies is $\mathcal{C}$ instead of the lattice of
projections. So it seems to be adequate to study structures which
include improper mixtures as well as pure states in a same level of
``discourse". Such structures could provide a natural framework in
which we study foundational issues related to these topics.

Let us see examples of physical situations which could be captured
by propositional structures based on $\mathcal{C}$. Suppose that we
have a system $S_{1}$ in a given state $\rho_{1}$ (which can be an
improper mixture). If we consider its environment $S_{2}$, then we
may state the proposition ``the state of affairs is such that
$S_{1}$ is in state $\rho_{1}$". We note that when we look the
things from the point of view of the total system
$S=\mbox{System}+\mbox{Environment}$, a convex subset of
$\mathcal{C}$ (the convex set of states of $S$) corresponds to this
proposition. This is so because $S$ can be in any state $\rho$ such
that $\mbox{tr}_{2}(\rho)=\rho_{1}$, and this corresponds to the
convex set $\mbox{tr}^{-1}_{2}(\{\rho_{1}\})$ (see section
\ref{s:inversetaumap}). Similarly, we obtain the convex set
$\mbox{tr}^{-1}_{2}(\{\rho_{1}\})\cap\mbox{tr}^{-1}_{1}(\{\rho_{2}\})$
for the proposition $S_{1}$ is in state $\rho_{1}$ and $S_{2}$ is in
state $\rho_{2}$. This propositions represent the ignorance that we
have about the actual state of the whole system. A propositional
structure which includes propositions of this kind could be useful,
and more natural for the study of quantum information.

It is important to notice that propositions such as the one
represented by $\mbox{tr}^{-1}_{2}(\{\rho_{1}\})$ above cannot be
tested by yes-no experiments in general. Notwithstanding, they
represent actual states of affairs, and they can certainly be tested
making measures on correlations, quantum tomographies, etc.

As another example, consider the von Newmann entropy
$S(\rho)=-\mbox{tr}(\rho\ln(\rho))$. It has the following property
of concavity \cite{bengtssonyczkowski2006}

\begin{prop}\label{p:concavity}
If $\rho=\alpha\rho_{1}+(1-\alpha)\rho_{2}$, $0\leq\alpha\leq 1$, we
have

\begin{equation}
S(\rho)\geq\alpha S(\rho_{1})+(1-\alpha)S(\rho_{2})
\end{equation}

\end{prop}

\noindent Now consider the proposition ``the entropy of the system
is greater than $S_{0}$". To such a proposition -which has a very
definite physical meaning- there corresponds a convex subset of
$\mathcal{C}$. This is so, because if we consider the set

\begin{equation}
S_{\geq S_{0}}=\{\rho\in\mathcal{C}\,|\,S(\rho)\geq S_{0}\}
\end{equation}

\noindent and if $\rho_{1},\rho_{2}\in S_{\geq S_{0}}$, then any
convex combination $\rho=\alpha\rho_{1}+(1-\alpha)\rho_{2}$ -due to
the concavity property- will also belong to $S_{\geq S_{0}}$. This
example shows that there are propositions with a very clear physical
meaning which correspond to subsets of $\mathcal{C}$ instead of
subspaces of the Hilbert space.

%For classical systems we consider the propositional structure of
%\emph{subsets} of the \emph{set} of states. It seems that sets and
%functions between sets are the important objects of $CM$. It seems
%that convexity is very important in quantum mechanics. In $QM$ the
%set of states is a convex subset of an affine variety. So, it would
%be interesting -in order to compare $QM$ and $CM$- to develop a
%propositional structure of convex subsets of the convex set of
%states. This would permit to put both theories in a more closer
%(closer than $$QL$$) common ground in order to compare them.

We summarize below the desired properties for the lattices that we
are searching for, in order to solve the problems posed in this
section:

\begin{enumerate}

\item[$1$]
All physical states are included as atoms of the the new lattice.
Atoms and physical states are in one to one correspondence.

\item[$2$]
A state of the system will be the conjunction of all the actual
properties (i.e. elements of the structure). This means that actual
properties determine univocally the state of the system.

\item[$3$]
There exist projection functions which map \emph{all} states (atoms)
of the structure corresponding to the whole system $S$, to the
corresponding states (atoms) of its subsystems $S_{1}$ and $S_{2}$.
This assignation rule must be compatible with the physics of the
problem.

\item[$4$]
$\mathcal{L}_{v\mathcal{N}}$ is set theoretically included in the
new structure, in order to preserve physical properties in the
standard sense.

\item[$5$]
Given two propositions of the structure there must exist an
operation which yields a proposition which expresses the fact that
we can form mixtures of states.

\end{enumerate}

There is a trivial example which satisfies the conditions $1-4$
listed above, namely, the set of all subsets of $\mathcal{C}$, which
we denote $\mathcal{P}(\mathcal{C})$. Using set intersection as
conjunction and set union as disjunction, it is a boolean lattice.
If we fix an entanglement measure, consider the proposition ``the
system has such amount of entanglement" or given an entropy measure,
we can say ```the system has such amount of entropy" and so on. To
such propositions we can assign elements of
$\mathcal{P}(\mathcal{C})$, the set of all states which satisfy
those propositions. But the boolean ``or" defined by the set union
hides the fact that in quantum mechanics we can make superpositions
of states (principle \ref{e:superposition principle}) and that we
can mix states (principle \ref{e:mixing principle}). In this work we
search for structures which satisfy condition $5$. For that reason,
the lattice formed by $\mathcal{P}(\mathcal{C})$ (from now on
$\mathcal{L}_{\mathcal{B}}$) is not of our interest. It expresses
the almost trivial fact that we can make propositions such as ``the
states of $\mathcal{C}$ which make a given function to have such a
value" but it hides the radical differences between $QM$ and $CM$.

We can define -at least- two structures which satisfy the above
list. One of them which we call $\mathcal{L}$ (see section
\ref{s:New language} and \cite{extendedql}) is in close connection
with the lattice of subspaces of the space of hermitian matrixes. In
this work it plays the role of a technical step to reach
$\mathcal{L}_{\mathcal{C}}$ (section \ref{s:Convex lattice}), the
lattice formed by the convex subsets of $\mathcal{C}$. We show below
that the study of these structures sheds light on the study of
compound quantum systems, and provide a suitable (natural) language
for them, mainly because of they sort the problems possed above.
They show things that $\mathcal{L}_{v\mathcal{N}}$ hides, or in
other words, which are not expressed clearly. For example, given two
pure states $\rho_{1}=|\psi_{1}\rangle\langle\psi_{1}|$ and
$\rho_{2}=|\psi_{2}\rangle\langle\psi_{2}|$ we can apply the ``or"
operation of $\mathcal{L}_{v\mathcal{N}}$,
$\vee_{\mathcal{L}_{v\mathcal{N}}}$, which yields the linear
(closed) spam of $|\psi_{1}\rangle$ and $|\psi_{2}\rangle$. But we
can also consider the ``$\vee_{\mathcal{L}_{\mathcal{C}}}$"
operation (see section \ref{s:Convex lattice}), which yields all
statistical mixtures of the form
$\alpha\rho_{1}+(1-\alpha)\rho_{2}$. This operation is different
from linear combination (quantum superpositions), and is related to
the -non classical- mixing of states (improper mixing). This
``mixing" operation cannot be represented in $QL$ at the level of
$\mathcal{L}_{v\mathcal{N}}$ itself, i.e., it is not a lattice
operation, but it has to be represented at the level of statistical
mixtures (measures over $\mathcal{L}_{v\mathcal{N}}$).

It is important to notice that it is not the aim of this work to
replace the von Newmann lattice by these new structures, but to
stress its limitations for the problem of compound systems and to
define its domain of applicability. We adopt the point of view that
these constructions -including $\mathcal{L}_{v\mathcal{N}}$- yield
different complementary views of quantum systems. In the following
sections, we present $\mathcal{L}$ and $\mathcal{L}_{\mathcal{C}}$.

\section{The Lattice $\mathcal{L}$}\label{s:New language}

In this section we review (without proof) some results and
definitions of \cite{extendedql}. Let us define $G(\mathcal{A})$ as
the lattice associated to the pair $(\mathcal{A},\mbox{tr})$, where
$\mathcal{A}$ is considered as an $\mathbb{R}$-vector space and
$\mbox{tr}$ is the usual trace operator on $B(\mathcal{H})$, which
induces the scalar product $<A,B>=tr(A\cdot B)$
($\dim(\mathcal{H})<\infty$). The restriction to $\mathcal{A}$ of
$\mbox{tr}$, makes $\mathcal{A}$ into an $\mathbb{R}$-euclidean
vector space.

\begin{equation}
G(\mathcal{A}):=\{S\subset \mathcal{A}\,|\, S\mbox{ is a }
\mathbb{R}-\mbox{subspace}\}
\end{equation}

\noindent $G(\mathcal{A})$ is a modular, orthocomplemented, atomic
and complete lattice (not distributive, hence not a Boolean
algebra). Let $\mathcal{L}$ be the induced lattice in $\mathcal{C}$:

\begin{equation}
\mathcal{L}:=\{S\cap\mathcal{C}\,|\, S\in G(\mathcal{A})\}
\end{equation}

\noindent There are a lot of subspaces $S\in G(\mathcal{A})$ such
that $S\cap \mathcal{C}=S'\cap \mathcal{C}$, so for each
$L\in\mathcal{L}$ we choose as a representative the subspace with
the least dimension:

\begin{equation}
\min\{\dim_{ \mathbb{R} }(S)\,|\,L=S\cap\mathcal{C},\,S\in
G(\mathcal{A})\}
\end{equation}

\noindent Let $[S]=L$, being  $S\in G(\mathcal{A})$ an element of
the class $L$, then

\begin{equation}
S\cap\mathcal{C}\subseteq<S\cap\mathcal{C}>_{\mathbb{R}}\subseteq
S\Rightarrow
S\cap\mathcal{C}\cap\mathcal{C}\subseteq<S\cap\mathcal{C}>_{\mathbb{R}}\cap\mathcal{C}
\subseteq
S\cap\mathcal{C}\Rightarrow
\end{equation}

\begin{equation}
<S\cap\mathcal{C}>\cap\mathcal{C}=S\cap\mathcal{C}
\end{equation}

\noindent So $<S\cap\mathcal{C}>$ and $S$ are in the same class $L$.
Note that $<S\cap\mathcal{C}>\subseteq S$ and if $S$ is the subspace
with the least dimension, then $<S\cap\mathcal{C}>=S$. Also note
that the representative with least dimension is unique, because if
we choose $S'$ such that $S'\cap\mathcal{C}=S\cap\mathcal{C}$, then

\begin{equation}
S=<S\cap\mathcal{C}>=<S'\cap\mathcal{C}>=S'
\end{equation}

\noindent Finally, the representative of a class $L$ that we choose
is the unique $\mathbb{R}$-subspace $S\subseteq\mathcal{A}$ such
that

\begin{equation}
S=<S\cap\mathcal{C}>_{\mathbb{R}}
\end{equation}

\noindent We call it the {\it good representative}. It is important
to remark that in the case of infinite dimensional Hilbert spaces we
cannot define good representatives in such a way.

Let us now define ``$\vee$", ``$\wedge$" and ``$\neg$" operations
and a partial ordering relation ``$\longrightarrow$" (or
equivalently ``$\leq$") in $\mathcal{L}$ as:

\begin{equation}
(S\cap\mathcal{C}) \wedge (T\cap\mathcal{C}):=
<S\cap\mathcal{C}>\cap<T\cap\mathcal{C}>\cap\mathcal{C}
\end{equation}

\begin{equation}
(S\cap\mathcal{C}) \vee (T\cap\mathcal{C}):=
(<S\cap\mathcal{C}>+<T\cap\mathcal{C}>)\cap\mathcal{C}
\end{equation}

\begin{equation}
(S\cap\mathcal{C}) \longrightarrow
(T\cap\mathcal{C})\Longleftrightarrow
(S\cap\mathcal{C})\subseteq(T\cap\mathcal{C})
\end{equation}

\begin{equation}
\neg(S\cap\mathcal{C}):=<S\cap\mathcal{C}>^\perp\cap\mathcal{C}
\end{equation}

With these operations, we have that

\begin{prop}
$\mathcal{L}$ is an atomic and complete lattice. If
$\dim(\mathcal{H})<\infty$, $\mathcal{L}$ is a modular lattice.
\end{prop}

$\mathcal{L}$ is not an orthocomplemented lattice, but it is easy to
show that non-contradiction holds

\begin{equation}
L\wedge\neg L=\mathbf{0}
\end{equation}

\noindent and also contraposition

\begin{equation}
L_{1}\leq L_{2}\Longrightarrow \neg L_{2}\leq L_{1}
\end{equation}

The following proposition links atoms and states:

\begin{prop}
There is a one to one correspondence between the states of the
system and the atoms of $\mathcal{L}$.
\end{prop}

It is well known \cite{bengtssonyczkowski2006} that there is a
lattice isomorphism between the complemented and complete lattice of
faces of the convex set $\mathcal{C}$ and
$\mathcal{L}_{v\mathcal{N}}$. Due to the following proposition

\begin{prop}
Every face of $\mathcal{C}$ is an element of $\mathcal{L}$.
\end{prop}

\noindent we conclude that

\begin{coro}
The complete lattice of faces of the convex set $\mathcal{C}$ is a
subposet of $\mathcal{L}$.
\end{coro}

The previous Corollary shows that $\mathcal{L}$ and
$\mathcal{L}_{v\mathcal{N}}$ are connected. What is the relationship
between their operations? If $F_{1}$ and $F_{2}$ are faces we have:

\begin{enumerate}

\item[($\wedge$)]
$F_1,F_2\in\mathcal{L}_{v\mathcal{N}}$, then $F_1\wedge F_2$ in
$\mathcal{L}_{v\mathcal{N}}$ is the same as in $\mathcal{L}$. So the
inclusion $\mathcal{L}_{v\mathcal{N}}\subseteq\mathcal{L}$ preserves
the $\wedge$-operation.

\item[($\vee$)]
$F_1\vee_{\mathcal{L}} F_2\leq F_1\vee_{\mathcal{L}_{v\mathcal{N}}}
F_2$ and $F_1\leq F_2\Rightarrow F_1\vee_{\mathcal{L}}
F_2=F_1\vee_{\mathcal{L}_{v\mathcal{N}}} F_2=F_2$

\item[$(\neg$)]
$\neg_{\mathcal{L}}F\leq \neg_{\mathcal{L}_{v\mathcal{N}}}F$
\end{enumerate}

Given two systems with Hilbert spaces $\mathcal{H}_{1}$ and
$\mathcal{H}_{2}$, we can construct the lattices $\mathcal{L}_{1}$
and $\mathcal{L}_{2}$. We can also construct $\mathcal{L}$, the
lattice associated to the product space
$\mathcal{H}_{1}\otimes\mathcal{H}_{2}$. We define:

\begin{equation}
\Psi:\mathcal{L}_{1}\times\mathcal{L}_{2}\longrightarrow\mathcal{L}\quad|\quad(S_{1}\cap\mathcal{C}_{1},S_{2}\cap\mathcal{C}_{2})\longrightarrow
S\cap\mathcal{C}
\end{equation}

\noindent where
$S=(<S_{1}\cap\mathcal{C}_{1}>\otimes<S_{2}\cap\mathcal{C}_{2}>)$.
In terms of good representatives, $\Psi([S_1],[S_2])=[S_1\otimes
S_2]$. An equivalent way to define it (in the finite dimensional
case) is by saying that $\Psi$ is the induced morphism in the
quotient lattices of the tensor map

\begin{equation}
G(\mathcal{A}_1)\times G(\mathcal{A}_2)\rightarrow
G(\mathcal{A}_1\otimes_{\mathbb{R}} \mathcal{A}_2)\cong
G(\mathcal{A})
\end{equation}

\noindent We can prove the following:

\begin{prop}\label{bimorfism}
Fixing $[U]\in\mathcal{L}_2$ then $\mathcal{L}_{1}$ is isomorphic
(as complete lattice) to $\mathcal{L}_{1}\times[U] \subseteq
\mathcal{L}$. The same is true for $\mathcal{L}_{2}$ and an
arbitrary element of $\mathcal{L}_1$.
\end{prop}\label{bi-mor}

Given $L_{1}\in\mathcal{L}_{1}$ and $L_{2}\in\mathcal{L}_{2}$, we
can define the following convex tensor product:

\begin{definition}\label{d:convex tensor product}
$L_{1}\widetilde{\otimes}\,L_{2}:=\{\sum\lambda_{ij}\rho_{i}^{1}\otimes\rho_{j}^{2}\,|\,\rho_{i}^{1}
\in L_{1},\,\,\rho_{j}^{2}\in L_{2},\,\, \sum\lambda_{ij}=1
\,\,\mbox{and} \,\,\lambda_{ij}\geq 0\}$
\end{definition}

\noindent This product is formed by all possible convex combinations
of tensor products of elements of $L_{1}$ and elements of $L_{2}$,
and it is again a convex set. Let us compute
$\mathcal{C}_{1}\widetilde{\otimes}\,\mathcal{C}_{2}$. Remember that
$\mathcal{C}_{1}=[\mathcal{A}_{1}]\in\mathcal{L}_{1}$ and
$\mathcal{C}_{2}=[\mathcal{A}_{2}]\in\mathcal{L}_{2}$:

\begin{equation}
\mathcal{C}_{1}\widetilde{\otimes}\,\mathcal{C}_{2}=
\{\sum\lambda_{ij}\rho_{i}^{1}\otimes\rho_{j}^{2}\,|\,\rho_{i}^{1}
\in \mathcal{C}_{1},\,\,\rho_{j}^{2}\in \mathcal{C}_{2},\,\,
\sum\lambda_{ij}=1 \,\,\mbox{and} \,\,\lambda_{ij}\geq 0\}
\end{equation}

\noindent So, if $\mathcal{S(\mathcal{H})}$ is the set of all
separable states, we have by definition:

\begin{equation}\label{e:separablestates}
\mathcal{S}(\mathcal{H})=\mathcal{C}_{1}\widetilde{\otimes}\,\mathcal{C}_{2}
\end{equation}

If the whole system is in a state $\rho$, using partial traces we
can define states for the subsystems $\rho_{1}=tr_{2}(\rho)$ and a
similar definition for $\rho_{2}$. Then, we can consider the maps:

\begin{equation}
\mbox{tr}_{i}:\mathcal{C}\longrightarrow \mathcal{C}_{j} \quad|\quad
\rho\longrightarrow \mbox{tr}_{i}(\rho)
\end{equation}

\noindent from which we can construct the induced projections:

\begin{equation}
\tau_{i}:\mathcal{L}\longrightarrow \mathcal{L}_{i} \quad|\quad
S\cap\mathcal{C}\longrightarrow \mbox{tr}_{i}( <S\cap\mathcal{C}>
)\cap \mathcal{C}_i
\end{equation}

In terms of good representatives $\tau_i([S])=[\mbox{tr}_i(S)]$.
Then we can define the product map

\begin{equation}
\tau:\mathcal{L}\longrightarrow\mathcal{L}_{1}\times\mathcal{L}_{2}
\quad|\quad L\longrightarrow(\tau_{1}(L),\tau_{2}(L))
\end{equation}

\noindent The maps defined in this section are shown in Figure
\ref{f:maps}.

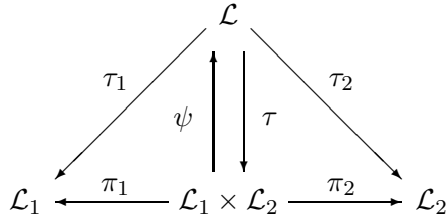
\begin{figure}\label{f:maps}
\begin{center}
\unitlength=1mm
\begin{picture}(5,5)(0,0)
\put(-3,23){\vector(-1,-1){20}} \put(3,23){\vector(1,-1){20}}
\put(-2,4){\vector(0,2){16}} \put(2,20){\vector(0,-2){16}}
\put(8,0){\vector(3,0){15}} \put(-8,0){\vector(-3,0){15}}

\put(0,25){\makebox(0,0){$\mathcal{L}$}}
\put(-27,0){\makebox(0,0){${\mathcal{L}_{1}}$}}
\put(27,0){\makebox(0,0){${\mathcal{L}_{2}}$}}
\put(0,0){\makebox(0,0){${\mathcal{L}_{1}\times\mathcal{L}_{2}}$}}
\put(-1,11){\makebox(-10,0){$\psi$}}
\put(-1,11){\makebox(13,0){$\tau$}}
\put(-15,16){\makebox(0,0){$\tau_{1}$}}
\put(15,16){\makebox(0,0){$\tau_{2}$}}
\put(-15,2){\makebox(0,0){$\pi_{1}$}}
\put(15,2){\makebox(0,0){$\pi_{2}$}}
\end{picture}
\caption{The different maps between $\mathcal{L}_{1}$,
$\mathcal{L}_{2}$, ${\mathcal{L}_{1}\times\mathcal{L}_{2}}$, and
$\mathcal{L}$. $\pi_{1}$ and $\pi_{2}$ are the canonical
projections.}\label{f:maps}
\end{center}

\end{figure}

\section{The Lattice of Convex Subsets}\label{s:Convex lattice}

The elements of $\mathcal{L}$ are formed by intersections between
closed subspaces and $\mathcal{C}$. Given that closed subspaces are
closed sets and so is $\mathcal{C}$, they are also convex subsets of
$\mathcal{C}$. We could go on further and consider all convex
subsets of $\mathcal{C}$. On the other hand (because of linearity),
partial trace operators preserve convexity and so they will map
propositions of the system into propositions of the subsystem, as
desired.

Another motivation for a further extension comes from the following
analogy. If the propositions of classical mechanics are the subsets
of $\Gamma$ (classical phase space), why cannot we consider the
\emph{convex} subsets of the \emph{convex} set of states? It seems,
after all, that convexity is an important feature of quantum
mechanics \cite{MielnikGQS}, \cite{MielnikTF}, and
\cite{MielnikGQM}. And as will be seen below, this
``convexification" of the lattice, allows for an algebraic
characterization of entanglement.

Let us begin by considering the set of all convex subsets of
$\mathcal{C}$:

\begin{definition}
$\mathcal{L}_{\mathcal{C}}:=\{C\subseteq\mathcal{C}\,|\, \mbox{C is
a convex subset of} \,\,\,\mathcal{C}\}$
\end{definition}

In order to give $\mathcal{L}_{\mathcal{C}}$ a lattice structure, we
introduce the following operations (where $conv(A)$ stands for
convex hull of a given set $A$):

\begin{definition}\label{definitionlattice}

For all $C,C_1,C_2\in\mathcal{L}_{\mathcal{C}}$

\begin{enumerate}

\item[$\wedge$]
$C_1\wedge C_2:= C_1\cap C_2$

\item[$\vee$]
$C_1\vee C_2:=conv(C_1,C_2)$. It is again a convex set, and it is
included in $\mathcal{C}$ (using convexity).

\item[$\neg$]
$\neg C:=C^{\perp}\cap\mathcal{C}$

\item[$\longrightarrow$]
$C_1\longrightarrow C_2:= C_1\subseteq C_2$

\end{enumerate}

\end{definition}

With the operations of definition \ref{definitionlattice}, it is
apparent that $(\mathcal{L}_{\mathcal{C}};\longrightarrow)$ is a
poset. If we set $\emptyset=\mathbf{0}$ and
$\mathcal{C}=\mathbf{1}$, then,
$(\mathcal{L}_{\mathcal{C}};\longrightarrow;\mathbf{0};\emptyset=\mathbf{0})$
will be a bounded poset.

\begin{prop}
$(\mathcal{L}_{\mathcal{C}};\longrightarrow;\wedge;\vee)$ satisfies

\begin{enumerate}

\item[$(a)$]
$C_1\wedge C_1= C_1$

\item[$(b)$]
$C_1\wedge C_2=C_2\wedge C_1$

\item[$(c)$]
$C_1\vee C_2=C_2\vee C_1$

\item[$(d)$]
$\neg C:=C^{\perp}\cap\mathcal{C}$

\item[$(d)$]
$C_1\wedge (C_2\wedge C_3)=(C_1\wedge C_2)\wedge C_3$

\item[$(e)$]
$C_1\vee (C_2\vee C_3)=(C_1\vee C_2)\vee C_3$

\item[$(f)$]
$C_1\wedge (C_1\vee C_2)=C_1$

\item[$(g)$]
$C_1\vee (C_1\wedge C_2)=C_1$
\end{enumerate}

\end{prop}

\begin{proof}
$C_1\wedge C_1= C_1\cap C_1=C_1$, so we have $(a)$. $(b)$, $(c)$ and
$(d)$ are equally trivial. In order to prove $e$ we have that

\begin{equation}
C_1\vee (C_2\vee C_3)=conv(C_1,conv(C_2,C_3))
\end{equation}

\noindent Given that $conv(C_2,C_3)\subseteq
conv(C_1,conv(C_2,C_3))$, then,

\begin{equation}
C_1,C_2,C_3\subseteq conv(C_1,conv(C_2,C_3))
\end{equation}

\noindent Using the above equation and convexity of
$conv(C_1,conv(C_2,C_3))$, we have that

\begin{equation}
conv(C_1,C_2)\subseteq conv(C_1,conv(C_2,C_3))
\end{equation}

\noindent and so, using convexity,

\begin{equation}
conv(conv(C_1,C_2),C_3)\subseteq conv(C_1,conv(C_2,C_3))
\end{equation}

\noindent A similar argument implies the converse inclusion, and so
we conclude that

\begin{equation}
(C_1\vee C_2)\vee
C_3=conv(conv(C_1,C_2),C_3)=conv(C_1,conv(C_2,C_3))=C_1\vee (C_2\vee
C_3)
\end{equation}

\noindent In order to prove $(f)$, we have $C_1\wedge (C_1\vee
C_2)=C_1\cap conv(C_1,C_2)$. As $C_1\cap conv(C_1,C_2)\subseteq C_1$
and $C_1\subseteq conv(C_1,C_2)$, we have $C_1=C_1\cap
conv(C_1,C_2)$, and so $(f)$ is true. Let us finally check $(g)$.
$C_1\vee (C_1\wedge C_2)=conv(C_1,C_1\cap C_2)$. This implies that
$C_1,C_1\cap C_2\subseteq conv(C_1,C_1\cap C_2)$. As $C_1$ is
convex, we have $conv(C_1,C_1\cap C_2)\subseteq C_1$, and so we have
$(g)$.
\end{proof}

Regarding the ``$\neg$" operation, if $C_1\subseteq  C_2$, then
$C_2^{\perp}\subseteq  C_1^{\perp}$. So $C_2^{\perp}\cap
\mathcal{C}\subseteq C_1^{\perp}\cap\mathcal{C}$, and hence

\begin{equation}
C_1\longrightarrow  C_2\Longrightarrow \neg C_2\longrightarrow \neg
C_1
\end{equation}

\noindent Given that $C\cap(C^{\perp}\cap\mathcal{C})=\emptyset$, we
also have:

\begin{equation}
C\wedge(\neg C)=\mathbf{0}
\end{equation}

\noindent and so, contraposition and non contradiction hold. But if
we take the proposition $C=\{\frac{1}{N}\mathrm{1}\}$, then an easy
calculation yields $\neg C=\mathbf{0}$. And then, $\neg(\neg
C)=\mathbf{1}$, and thus $\neg(\neg C)\neq C$ in general. Double
negation does not hold, thus, $\mathcal{L}_{\mathcal{C}}$ is not an
ortholattice.

$\mathcal{L}_{\mathcal{C}}$ is a lattice which includes all convex
subsets of the quantum space of states. It includes $\mathcal{L}$,
and so, all quantum states (including all improper mixtures) as
propositions. It is also in strong analogy with classical physics,
where the lattice of propositions is formed by all measurable
subsets of phase space (the space of states).

\subsection{The Relationship Between $\mathcal{L}_{v\mathcal{N}}$, $\mathcal{L}$ and
$\mathcal{L}_{\mathcal{C}}$}\label{s:The Relationship}

\begin{prop}\label{p:Inclusion of Lvn}
$\mathcal{L}_{v\mathcal{N}}\subseteq\mathcal{L}\subseteq\mathcal{L}_{\mathcal{C}}$
as posets.
\end{prop}

\begin{proof}
We have already seen that
$\mathcal{L}_{v\mathcal{N}}\subseteq\mathcal{L}$ as sets. Moreover
it is easy to see that if $F_1\leq F_2$ in
$\mathcal{L}_{v\mathcal{N}}$ then $F_1\leq F_2$ in $\mathcal{L}$.
This is so because both orders are set theoretical inclusions.
Similarly, if $L_{1},L_{2}\in\mathcal{L}$, because intersection of
convex sets yields a convex set (and closed subspaces are convex
sets also), $L_{1},L_{2}\in\mathcal{L}_{\mathcal{C}}$, then we
obtain set theoretical inclusion. And, again, because of both orders
are set theoretical inclusions, we obtain that they are included as
posets.
\end{proof}

Regarding the $\vee$ operation, let us compare
$\vee_{\mathcal{L}_{v\mathcal{N}}}$, $\vee_{\mathcal{L}}$ and
$\vee_{\mathcal{L}_{C}}$. If $L_{1},L_{2}\in\mathcal{L}$, then they
are convex sets and so, $L_{1},L_{2}\in\mathcal{L}_{\mathcal{C}}$.
Then we can compute

\begin{equation}
L_{1}\vee_{\mathcal{L}_{C}}L_{2}=conv(L_{1},L_{2})
\end{equation}

\noindent On the other hand (if $S_{1}$ and $S_{2}$ are good
representatives for $L_{1}$ and $L_{2}$), then:

\begin{equation}
L_{1}\vee_{\mathcal{L}}L_{2}=(<S_{1}\cap\mathcal{C}>+<S_{2}\cap\mathcal{C}>)\cap\mathcal{C}
\end{equation}

\noindent The direct sum of the subspaces $<S_{1}\cap\mathcal{C}>$
and $<S_{2}\cap\mathcal{C}>$ contains as a particular case all
convex combinations of elements of $L_{1}$ and $L_{2}$. So we can
conclude

\begin{equation}
L_{1}\vee_{\mathcal{L}_{C}}L_{2}\leq L_{1}\vee_{\mathcal{L}}L_{2}
\end{equation}

As faces of $\mathcal{C}$ can be considered as elements of
$\mathcal{L}_{C}$ because they are convex, if $F_{1}$ and $F_{2}$
are faces, we can also state

\begin{equation}
F_{1}\vee_{\mathcal{L}_{C}}F_{2}\leq
F_{1}\vee_{\mathcal{L}}F_{2}\leq
F_{1}\vee_{\mathcal{L}_{v\mathcal{N}}}F_{2}
\end{equation}

Intersection of convex sets is the same as intersection of elements
of $\mathcal{L}$ and so we have

\begin{equation}
L_{1}\wedge_{\mathcal{L}_{C}}L_{2}= L_{1}\wedge_{\mathcal{L}}L_{2}
\end{equation}

\noindent and similarly

\begin{equation}
F_{1}\wedge{\mathcal{L}_{v\mathcal{N}}}F_{2}=
F_{1}\wedge_{\mathcal{L}_{C}}F_{2}= F_{1}\wedge_{\mathcal{L}}F_{2}
\end{equation}

What is the relationship between $\neg_{\mathcal{L}_{\mathcal{C}}}$
and $\neg_{\mathcal{L}}$? Suppose that $L_{1}\in\mathcal{L}$, then
they are convex sets also, and so
$L_{1}\in\mathcal{L}_{\mathcal{C}}$. Then we can compute
$\neg_{\mathcal{L}_{\mathcal{C}}}L_{1}$. We obtain:

\begin{equation}
\neg_{\mathcal{L}_{\mathcal{C}}}L_{1}=L_{1}^{\perp}\cap\mathcal{C}
\end{equation}

\noindent On the other hand, if $L_{1}=S\cap\mathcal{C}$, with $S$ a
good representative

\begin{equation}
\neg_{\mathcal{L}}L_{1}=<S\cap\mathcal{C}>^{\perp}\cap\mathcal{C}
\end{equation}

\noindent As $L_{1}\subseteq <S\cap\mathcal{C}>$, then
$<S\cap\mathcal{C}>^{\perp}\subseteq L_{1}^{\perp}$, and so

\begin{equation}
\neg_{\mathcal{L}}L_{1}\leq\neg_{\mathcal{L}_{\mathcal{C}}}L_{1}
\end{equation}

\subsection{Interactions in $QM$ and $CM$ Compared}\label{s:interaction in QM}

The origin of the extension of $\mathcal{L}_{v\mathcal{N}}$ becomes
clear if wee make a comparison between classical and quantum
compound systems. For a single classical system its properties are
faithfully represented by the subsets of its phase space. When
another classical system is added and the compound system is
considered, no enrichment of the state space of the former system is
needed in order to describe its properties, \emph{even} in the
presence of interactions. No matter which the interactions may be,
the cartesian product of phase spaces is sufficient for the
description of the compound system.

The situation is quite different in quantum mechanics. This is so
because, if we add a new quantum system to a previously isolated
one, pure states are no longer faithful in order to describe
subsystems. Interactions produce non trivial correlations, which are
reflected in the presence of entangled states (and violation of Bell
inequalities). Thus, we have to consider the information about the
non trivial correlations that each subsystem has with other
subsystems -for example, a system with the environment. The
existence of this additional information implies that we must add
new elements to the propositional structure of the system.

\section{The Relationship Between $\mathcal{L}_{\mathcal{C}}$ and The Tensor
Product of Hilbert Spaces}\label{s:The Relationship for convex}

In this section we study the relationship between the lattice
$\mathcal{L}_{\mathcal{C}}$ of a system $S$ composed of subsystems
$S_{1}$ and $S_{2}$, and the lattices of its subsystems,
$\mathcal{L}_{\mathcal{C}1}$ and $\mathcal{L}_{\mathcal{C}2}$
respectively. As in \cite{extendedql}, we do this by making the
physical interpretation of maps which can be defined between them.

\subsection{Separable States (Going Up)}\label{s:going up}

Let us define:

\begin{definition}
Given $C_{1}\subseteq\mathcal{C}_{1}$ and
$C_{2}\subseteq\mathcal{C}_{2}$

\begin{equation}
C_1\otimes C_2:=\{\rho_{1}\otimes\rho_{2}\,|\,\rho_{1}\in
C_1,\rho_{2}\in C_2\}
\end{equation}

\end{definition}

\noindent Then, we define the map:

\begin{definition}
$$\Lambda:\mathcal{L}_{\mathcal{C}1}\times\mathcal{L}_{\mathcal{C}2}\longrightarrow\mathcal{L}_{\mathcal{C}}$$
$$(C_{1},C_{2})\longrightarrow conv(C_1\otimes C_2)$$
\end{definition}

\noindent In the rest of this work will use the following
proposition (see for example \cite{Convexsets}):

\begin{prop}
Let $S$ be a subset of a linear space $\mathbb{L}$. Then $x\in
conv(S)$ iff x is contained in a finite dimensional simplex $\Delta$
whose vertices belong to $S$.
\end{prop}

From equation \ref{e:separablestates} and definition \ref{d:convex
tensor product} it should be clear that
$\Lambda(\mathcal{C}_{1},\mathcal{C}_{2})=\mathcal{S}(\mathcal{H})$.
Definition \ref{d:convex tensor product} also implies that for all
$C_{1}\subseteq\mathcal{C}_{1}$ and $C_{2}\subseteq\mathcal{C}_{2}$:

\begin{equation}
\Lambda(C_{1},C_{2})=C_{1}\widetilde{\otimes}\,C_{2}
\end{equation}

\begin{prop}
Let $\rho=\rho_{1}\otimes\rho_{2}$, with
$\rho_{1}\in\mathcal{C}_{1}$ and $\rho_{2}\in\mathcal{C}_{2}$. Then
$\{\rho\}=\Lambda(\{\rho_{1}\},\{\rho_{2}\})$ with
$\{\rho_1\}\in\mathcal{L}_{C1}$, $\{\rho_2\}\in\mathcal{L}_{C2}$ and
$\{\rho\}\in\mathcal{C}$.
\end{prop}

\begin{proof}
We already know that the atoms are elements of the lattices. Thus,

\begin{equation}
\Lambda(\{\rho_{1}\},\{\rho_{2}\})
=conv(\{\rho_{1}\otimes\rho_{2}\})=\{\rho_{1}\otimes\rho_{2}\}=\{\rho\}
\end{equation}

\end{proof}

\begin{prop}
Let $\rho\in\mathcal{S(\mathcal{H})}$, the set of separable states.
Then, there exist $C\in\mathcal{L}_{\mathcal{C}}$,
$C_{1}\in\mathcal{L}_{\mathcal{C}_{1}}$ and
$C_{2}\in\mathcal{L}_{C_{2}}$ such that $\rho\in C$ and
$L=\Lambda(C_{1},C_{2})$.
\end{prop}

\begin{proof}
If $\rho\in\mathcal{S(\mathcal{H})}$, then
$\rho=\sum_{ij}\lambda_{ij}\rho_{i}^{1}\otimes\rho_{j}^{2}$, with
$\sum_{ij}\lambda_{ij}=1$ and $\lambda_{ij}\geq 0$. Consider the
convex sets:

\begin{equation}
C_{1}=conv(\{\rho_{1}^{1},\rho_{2}^{1},\cdots,\rho_{k}^{1}\})\quad
C_{2}=conv(\{\rho_{1}^{2},\rho_{2}^{2},\cdots,\rho_{l}^{2}\})
\end{equation}

\noindent Then we define:

$$C:=\Lambda(C_{1},C_{2})=conv(C_{1}\otimes C_{2})$$

\noindent Clearly, the set
$\{\rho_{i}^{1}\otimes\rho_{j}^{2}\}\subseteq C_{1}\otimes C_{2}$,
and then $\rho\in C$.
\end{proof}

\subsection{Projections Onto $\mathcal{L}_{\mathcal{C}_{1}}$ and $\mathcal{L}_{\mathcal{C}_{2}}$ (Going
Down)}\label{s:projectionsc}

Let us now study the projections onto
$\mathcal{L}_{\mathcal{C}_{1}}$ and $\mathcal{L}_{\mathcal{C}_{2}}$.
From a physical point of view, it is of interest to study the
partial trace operators. If the whole system is in a state $\rho$,
using partial traces we can define states for the subsystems
$\rho_{1}=tr_{2}(\rho)$ and a similar definition for $\rho_{2}$.
Then, we can consider the maps:

\begin{equation}
\mbox{tr}_{i}:\mathcal{C}\longrightarrow \mathcal{C}_{j} \quad|\quad
\rho\longrightarrow \mbox{tr}_{i}(\rho)
\end{equation}

\noindent from which we can construct the induced projections:

\begin{equation}
\tau_{i}:\mathcal{L}_{\mathcal{C}}\longrightarrow
\mathcal{L}_{\mathcal{C}_{i}} \quad|\quad C\longrightarrow
\mbox{tr}_{i}( C )
\end{equation}

\noindent Then we can define the product map

\begin{equation}
\tau:\mathcal{L}_{\mathcal{C}}\longrightarrow\mathcal{L}_{\mathcal{C}_{1}}\times\mathcal{L}_{\mathcal{C}_{2}}
\quad|\quad C\longrightarrow(\tau_{1}(C),\tau_{2}(C))
\end{equation}

We use the same notation for $\tau$ and $\tau_{i}$ (though they are
different functions) as in \cite{extendedql} and section \ref{s:New
language}, and this should not introduce any difficulty. We can
prove the following about the image of $\tau_{i}$.

\begin{prop}\label{lastausonsurjective}
The maps $\tau_{i}$ preserve the convex structure, i.e., they map
convex sets into convex sets.
\end{prop}

\begin{proof}
Let $C\subseteq \mathcal{C}$ be a convex set. Let $C_{1}$ be the
image of $C$ under $\tau_{2}$ (a similar argument holds for
$\tau_{1}$). Let us show that $C_{1}$ is convex. Let $\rho_{1}$ and
$\rho'_{1}$ be elements of $C_{1}$. Consider
$\sigma_{1}=\alpha\rho_{1}+(1-\alpha)\rho'_{1}$, with
$0\leq\alpha\leq 1$. Then, there exists $\rho,\rho'\in\mathcal{C}$
such that:

\begin{equation}
\sigma_{1}=\alpha\mbox{tr}_2(\rho)+(1-\alpha)\mbox{tr}_2(\rho')=\mbox{tr}_{2}(\alpha\rho+(1-\alpha)\rho')
\end{equation}

\noindent where we have used the linearity of trace. Because of
convexity of $C$, $\sigma:=\alpha\rho+(1-\alpha)\rho'\in C$, and so,
$\sigma_{1}=\mbox{tr}_{2}(\sigma)\in C_{1}$.
\end{proof}

\begin{prop}\label{lastausonsurjective}
The functions $\tau_{i}$ are surjective and preserve the
$\vee$-operation. They are not injective.
\end{prop}

\begin{proof}
Take the convex set $C_{1}\in\mathcal{L}_{\mathcal{C}_{1}}$. Choose
an arbitrary element of $\mathcal{C}_{2}$, say $\rho_{2}$. Now
consider the following element of $\mathcal{L}_{\mathcal{C}}$

\begin{equation}
C=C_{1}\otimes\rho_{2}
\end{equation}

\noindent $C$ is convex, and so belongs to
$\mathcal{L}_{\mathcal{C}}$, because if
$\rho\otimes\rho_{2},\sigma\otimes\rho_{2}\in C$, then any convex
combination
$\alpha\rho\otimes\rho_{2}+(1-\alpha)\sigma\otimes\rho_{2}=(\alpha\rho+(1-\alpha)\sigma)\otimes\rho_{2}\in
C$ (where we have used convexity of $C_{1}$). It is clear that
$\tau_{1}(C)=C_{1}$, because if $\rho_{1}\in C_{1}$, then
$\mbox{tr}(\rho_{1}\otimes\rho_{2})=\rho_{1}$. So, $\tau_{1}$ is
surjective. On the other hand, the arbitrariness of $\rho_{2}$
implies that it is not injective. An analogous argument follows for $\tau_2$.\\
Let us see that $\tau_i$ preserves the $\vee$-operation. Let $C$ and
$C'$ be convex subsets of $\mathcal{C}$. We must compute
$\mbox{tr}_{2}(C\vee C'))=\mbox{tr}_{2}(conv(C,C'))$. We must show
that this is the same as $conv(\mbox{tr}_{2}(C),\mbox{tr}_{2}(C'))$.
Take $x\in conv(\mbox{tr}_{2}(C),\mbox{tr}_{2}(C'))$. Then
$x=\alpha\mbox{tr}_{2}(\rho)+(1-\alpha)\mbox{tr}_{2}(\rho')$, with
$\rho\in C$, $\rho'\in C'$ and $0\leq\alpha\leq 1$. Using linearity
of trace, $x=\mbox{tr}_{2}(\alpha\rho+(1-\alpha)\rho')$.
$\alpha\rho+(1-\alpha)\rho'\in conv(C,C')$, and so,
$x\in\mbox{tr}_{2}(conv(C,C'))$. Hence we have

\begin{equation}
conv(\mbox{tr}_{2}(C),\mbox{tr}_{2}(C'))\subseteq\mbox{tr}_{2}(conv(C,C'))
\end{equation}

\noindent In other to prove the other inclusion, take
$x\in\mbox{tr}_{2}(conv(C,C'))$. Then,

\begin{equation}
x=\mbox{tr}_{2}(\alpha\rho+(1-\alpha)\rho')=\alpha\mbox{tr}_{2}(\rho)+(1-\alpha)\mbox{tr}_{2}(\rho')
\end{equation}

\noindent with $\rho\in C_{1}$ and $\rho'\in C'$. On the other hand,
$\mbox{tr}_{2}(\rho)\in\mbox{tr}_{2}(C)$ and
$\mbox{tr}_{2}(\rho')\in\mbox{tr}_{2}(C')$. This proves that:

$$\mbox{tr}_{2}(conv(C,C'))\subseteq conv(\mbox{tr}_{2}(C),\mbox{tr}_{2}(C'))$$

\end{proof}

\noindent Let us now consider the $\wedge$-operation. If
$x\in\tau_i(C\wedge C')=\tau_i(C\cap C')$ then $x=\tau_i(\rho)$ with
$\rho\in C\cap C'$. But if $\rho\in C$, then
$x=\tau_i(\rho)\in\mbox{tr}_i(C)$. As $\rho\in C'$ also, a similar
argument shows that $x=\tau_i(\rho)\in\mbox{tr}_i(C')$. Then
$x\in\tau_i(C)\cap \tau_i(C')$. And so:

\begin{equation}
\tau_i(C\cap C')\subseteq\tau_i(C)\cap \tau_i(C')
\end{equation}

\noindent which is the same as:

\begin{equation}
\tau_i(C\wedge C')\leq\tau_i(C)\wedge\tau_i(C')
\end{equation}

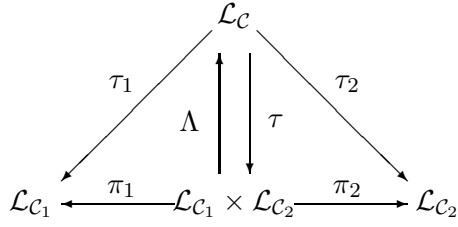
\begin{figure}\label{f:maps}
\begin{center}
\unitlength=1mm
\begin{picture}(5,5)(0,0)
\put(-3,23){\vector(-1,-1){20}} \put(3,23){\vector(1,-1){20}}
\put(-2,4){\vector(0,2){16}} \put(2,20){\vector(0,-2){16}}
\put(8,0){\vector(3,0){15}} \put(-8,0){\vector(-3,0){15}}

\put(0,25){\makebox(0,0){$\mathcal{L}_{\mathcal{C}}$}}
\put(-27,0){\makebox(0,0){${\mathcal{L}_{\mathcal{C}_{1}}}$}}
\put(27,0){\makebox(0,0){${\mathcal{L}_{\mathcal{C}_{2}}}$}}
\put(0,0){\makebox(0,0){${\mathcal{L}_{\mathcal{C}_{1}}\times\mathcal{L}_{\mathcal{C}_{2}}}$}}
\put(-1,11){\makebox(-10,0){$\Lambda$}}
\put(-1,11){\makebox(13,0){$\tau$}}
\put(-15,16){\makebox(0,0){$\tau_{1}$}}
\put(15,16){\makebox(0,0){$\tau_{2}$}}
\put(-15,2){\makebox(0,0){$\pi_{1}$}}
\put(15,2){\makebox(0,0){$\pi_{2}$}}
\end{picture}
\caption{The different maps between $\mathcal{L}_{\mathcal{C}_{1}}$,
$\mathcal{L}_{\mathcal{C}_{2}}$,
${\mathcal{L}_{\mathcal{C}_{1}}\times\mathcal{L}_{\mathcal{C}_{2}}}$,
and $\mathcal{L}_{\mathcal{C}}$}\label{f:maps}
\end{center}

\end{figure}

\noindent But these sets are not equal in general, as the following
example shows. Take $\{\rho_{1}\otimes\rho_{2}\}\in\mathcal{L}$ and
$\{\rho_{1}\otimes\rho'_{2}\}\in\mathcal{L}$, with $\rho'\neq\rho$.
It is clear that
$\{\rho_{1}\otimes\rho_{2}\}\wedge\{\rho_{1}\otimes\rho'_{2}\}=\mathbf{0}$
and so,
$\tau_{1}(\{\rho_{1}\otimes\rho_{2}\}\wedge\{\rho_{1}\otimes\rho'_{2}\})=\mathbf{0}$.
On the other hand,
$\tau_{1}(\{\rho_{1}\otimes\rho_{2}\})=\{\rho_{1}\}=\tau_{1}(\{\rho_{1}\otimes\rho'_{2}\})$,
and  so,
$\tau_{1}(\{\rho_{1}\otimes\rho_{2}\})\wedge\tau_{1}(\{\rho_{1}\otimes\rho'_{2}\})=\{\rho_{1}\}$.
A similar fact holds for the $\neg$-operation.

The last result is in strong analogy with what happens in
$\mathcal{L}$, where lack of injectivity of the $\tau_{i}$ may be
physically interpreted in the fact that the whole system has much
more information than that of its parts. It is again useful to make
a comparison with the classical case in order to illustrate what is
happening. The same as in classical mechanics, we have atoms in
$\mathcal{L}$ which are tensor products of atoms of
$\mathcal{L}_{1}$ and $\mathcal{L}_{2}$. But in contrast to
classical mechanics, entangled states originate atoms of
$\mathcal{L}$ which cannot be expressed in such a way, and thus, the
fiber of the projection $\tau_{i}$ is much bigger than that of its
classical counterpart.

It is again an important result that the projection function $\tau$
cannot be defined properly within the frame of the traditional
approaches of $QL$ because there is no place for improper mixtures
in those formalisms. But in the formalism presented here they are
included as \emph{elements} of the lattices, and so we can define
the projections from the lattice of the whole system to the lattices
of the subsystems. This enables a more natural approach when
compound systems are considered from a quantum logical point of
view.

\subsection{An Algebraic Characterization for
Entanglement}\label{s:entanglement}

We shown that it is possible to extend $\mathcal{L}_{v\mathcal{N}}$
in order to deal with statistical mixtures and that $\mathcal{L}$
and $\mathcal{L}_{\mathcal{C}}$ are possible extensions. It would be
interesting to search for a characterization of entanglement within
this framework. Let us see first what happens with the functions
$\Lambda\circ\tau$ and $\tau\circ\Lambda$. We have:

\begin{prop}\label{subir bajar}
$\tau\circ\Lambda=Id$.
\end{prop}

\begin{proof}
$$\tau_1(\Lambda(C_{1},C_{2}))=\tau_1(conv(C_{1}\otimes C_{2}))=\mbox{tr}_2(conv(C_{1}\otimes C_{2}))=C_{1}$$
$$\tau_2(\Lambda(C_{1},C_{2}))=\tau_2(conv(C_{1}\otimes C_{2}))=\mbox{tr}_1(conv(C_{1}\otimes C_{2}))=C_{2}$$
Then $\tau(\Lambda(C_{1},C_{2}))=(C_{1},C_{2})$.
\end{proof}

Again, as in \cite{extendedql}, if we take into account physical
considerations, $\Lambda\circ\tau$ is not the identity function.
This is because when we take partial traces, we face the risk of
losing information which will not be recovered when we make products
of states. So we obtain the same slogan as before \cite{extendedql}:
\emph{``going down and then going up is not the same as going up and
then going down''}. We show these maps in Figure \ref{f:maps}. How
is this related to entanglement? If we restrict $\Lambda\circ\tau$
to the set of product states, then it reduces to the identity
function, for if $\rho=\rho_{1}\otimes\rho_{2}$, then:

\begin{equation}
\Lambda\circ\tau(\{\rho\})=\{\rho\}
\end{equation}

\noindent On the other hand, it should be clear that if $\rho$ is an
entangled state

\begin{equation}\label{entangled equation}
\Lambda\circ\tau(\{\rho\})\neq\{\rho\}
\end{equation}

\noindent because
$\Lambda\circ\tau(\{\rho\})=\{\mbox{Tr}_{2}(\rho)\otimes\mbox{Tr}_{1}(\rho)\}\neq\{\rho\}$
for any entangled state. This property points in the direction of an
arrow characterization of entanglement. There are mixed states which
are not product states, and so, entangled states are not the only
ones who satisfy equation \ref{entangled equation}. What is the
condition satisfied for a general mixed state? The following
proposition summarizes all of this.

\begin{prop}\label{subir bajar}
If $\rho$ is a separable state, then there exists a convex set
$S_{\rho}\subseteq\mathcal{S}(\mathcal{H})$ such that $\rho\in
S_{\rho}$ and $\Lambda\circ\tau(S_{\rho})=S_{\rho}$. More generally,
for a convex set $C\subseteq \mathcal{S}(\mathcal{H})$, then there
exists a convex set $S_{C}\subseteq\mathcal{S}(\mathcal{H})$ such
that $\Lambda\circ\tau(S_{C})=S_{C}$. For a product state, we can
choose $S_{\rho}=\{\rho\}$. Any proposition
$C\in\mathcal{L}_{\mathcal{C}}$ which has at lest one non-separable
state, satisfies that there is no convex set $S$ such that
$C\subseteq S$ and $\Lambda\circ\tau(S)=S$.
\end{prop}

\begin{proof}
We have already seen above that if $\rho$ is a product state, then
$\Lambda\circ\tau(\{\rho\})=\{\rho\}$, and so $S_{\rho}=\{\rho\}$.
If $\rho$ is a general separable state, then there exists
$\rho_{1k}\in\mathcal{C}_{1}$, $\rho_{2k}\in\mathcal{C}_{1}$ and
$\alpha_{k}\geq 0, \sum_{k=1}^{N}\alpha_{k}=1$ such that
$\rho=\sum_{k=1}^{N}\alpha_{k}\rho_{1k}\otimes\rho_{2k}$. Now
consider the convex set (a simplex)

\begin{equation}
M=\{\sigma\in\mathcal{C}\,|\,\sigma=\sum_{i,j=1}^{N}\lambda_{i,j}\rho_{1i}\otimes\rho_{2j},
\lambda_{i,j}\geq 0, \sum_{i,j=1}^{N}\lambda_{i,j}=1\}
\end{equation}

\noindent It is formed by all convex combinations of products of the
elements which appear in the decomposition of $\rho$. It should be
clear that $\rho\in M$. If we apply
$(\mbox{tr}_{1}(),\mbox{tr}_{2}())$ to $\sigma\in M$, we get

\begin{equation}
(\mbox{tr}_{1}(\sigma),\mbox{tr}_{2}(\sigma))=(\sum_{i=1}^{N}(\sum_{j=1}^{N}\lambda_{i,j})
\rho_{1i},\sum_{j=1}^{N}(\sum_{i=1}^{N}\lambda_{i,j})\rho_{2j})=(\sum_{i=1}^{N}\mu_{i}
\rho_{1i},\sum_{j=1}^{N}\nu_{j}\rho_{2j})
\end{equation}

\noindent with $\mu_{i}=\sum_{j=1}^{N}\lambda_{i,j}$ and
$\nu_{i}=\sum_{i=1}^{N}\lambda_{i,j}$. Note that
$\sum_{j=1}^{N}\mu_{j}=\sum_{j=1}^{N}\nu_{j}=1$. If we now apply
$\Lambda$:

\begin{equation}
\Lambda((\sum_{i=1}^{N}\mu_{i}
\rho_{1i},\sum_{j=1}^{N}\nu_{j}\rho_{2j}))=\sum_{i,j=1}^{N}\mu_{i}\nu_{j}\rho_{1i}\otimes\rho_{2j}
\end{equation}

\noindent which is an element of $M$, and so, we conclude that
$\Lambda\circ\tau(M)\subseteq M$. On the other hand, if $\sigma\in
M$, then
$\sigma=\sum_{i,j=1}^{N}\lambda_{i,j}\rho_{1i}\otimes\rho_{2j}$
(convex combination). It is important to note that
$\Lambda\circ\tau(M)$ is a convex set, because trace operators
preserve convexity, and $\Lambda$ is a convex hull. On the other
hand
$\Lambda\circ\tau(\{\rho_{1i}\otimes\rho_{2j}\})=\{\rho_{1i}\otimes\rho_{2j}\}$.
And so, by convexity of $\Lambda\circ\tau(M)$,
$\sigma\in\Lambda\circ\tau(M)$. Finally, $\Lambda\circ\tau(M)=M$
(and $\rho\in M$). Then $M$ is the desired
$S_{\rho}\subseteq\mathcal{S}(\mathcal{H})$.

If $C\subseteq \mathcal{S}(\mathcal{H})$, then all $\rho\in C$ are
separable. $\mathcal{S}(\mathcal{H})$ is by definition, a convex
set. Let us see that it is invariant under $\Lambda\circ\tau$. First
of all, we know that $\mathcal{S}(\mathcal{H})$ is formed by all
possible convex combinations of the from $\rho_{1}\otimes\rho_{2}$,
with $\rho_{1}\in\mathcal{C}_{1}$ and $\rho_{2}\in\mathcal{C}_{2}$.
But for each one of these tensor products,
$\Lambda\circ\tau(\{\rho_{1}\otimes\rho_{2}\})=\{\rho_{1}\otimes\rho_{2}\}$,
and so belongs to $\Lambda\circ\tau(\mathcal{S}(\mathcal{H}))$. This
is a convex set, thus all convex combinations of them belong to it.
So we can conclude that

\begin{equation}
\Lambda\circ\tau(\mathcal{S}(\mathcal{H}))=\mathcal{S}(\mathcal{H})
\end{equation}

Now, consider $C\in\mathcal{L}_{\mathcal{C}}$ such that there exists
$\rho\in C$, being $\rho$ nonseparable.
$\Lambda\circ\tau(S)\subseteq \mathcal{S}(\mathcal{H})$ for all
$S\in\mathcal{L}_{\mathcal{C}}$. Then, it could never happen that
there exists $S\in\mathcal{L}_{\mathcal{C}}$ such that $C\subseteq
S$ and $\Lambda\circ\tau(S)=S$.
\end{proof}

From the last proposition, we conclude that there is a property
which the convex subsets of separable states satisfy, and convex
subsets which include non-separable states do not. This motivates
the following definition.

\begin{definition}
If $C\in\mathcal{L}_{\mathcal{C}}$, we will say that it is a
separable proposition if there exists
$S_{C}\in\mathcal{L}_{\mathcal{C}}$ such that
$\Lambda\circ\tau(S_{C})=S_{C}$ and $C\subseteq S$. Otherwise, we
will say that it is a non-separable or entangled proposition.
\end{definition}

\subsection{The Inverse $\tau$-map}\label{s:inversetaumap}

In section \ref{s:projectionsc} we defined the function
$\tau=(\tau_{1},\tau_{2})$. Now we show that using the inverse map
$\tau^{-1}=(\tau^{-1}_{1},\tau^{-1}_{2})$ we obtain lattice
morphisms. It is easy to show that $\tau^{-1}_{i}$ maps any
proposition from $\mathcal{C}_{i}$ into a proposition of
$\mathcal{C}$. This is because the pre-image of a convex set under
these functions is again a convex set. If $C_{1}$ is a proposition
of $\mathcal{C}$ and if $\tau_{1}(\rho),\tau_{1}(\rho')\in C_{1}$,
it is clear that any convex combination of $\rho$ and $\rho'$ will
belong to $\tau^{-1}_{1}(C_{1})$, because the partial trace is
linear and $C_{1}$ is convex.

%Consider then the following set

%\begin{equation}
%\Omega=\{(\tau^{-1}_{1}(C_{1}),\tau^{-1}_{2}(C_{2}))\in\mathcal{L}_{\mathcal{C}}\times\mathcal{L}_{\mathcal{C}}\,|\,
%C_{1}\in\mathcal{L}_{\mathcal{C}_{1}}\,\mbox{and}\,C_{2}\in\mathcal{L}_{\mathcal{C}_{2}}\}
%\end{equation}

%With the elements of $\Omega$ we can construct different maps.
%Consider
%$\Upsilon:\mathcal{L}_{\mathcal{C}_{1}}\times\mathcal{L}_{\mathcal{C}_{2}}\longrightarrow\mathcal{L}_{\mathcal{C}}$
%such that
%$\Upsilon(C_{1},C_{2})=\tau^{-1}_{1}(C_{1})\vee\tau^{-1}_{2}(C_{2})$.
%We can construct also the map
%$\Xi:\mathcal{L}_{\mathcal{C}_{1}}\times\mathcal{L}_{\mathcal{C}_{2}}\longrightarrow\mathcal{L}_{\mathcal{C}}$
%such that
%$\Xi(C_{1},C_{2})=\tau^{-1}_{1}(C_{1})\wedge\tau^{-1}_{2}(C_{2})$.
%We can prove the following propositions:

%\begin{prop}
%$\Xi(\mathcal{C}_{1},\mathcal{C}_{2})=\Upsilon(\mathcal{C}_{1},\mathcal{C}_{2})=\mathcal{C}$
%\end{prop}

%\begin{proof}
%It is easy to show that
%$\tau^{-1}_{1}(\mathcal{C}_{1})=\mathcal{C}=\tau^{-1}_{2}(\mathcal{C}_{2})$.
%Then, it follows that
%$\Xi(\mathcal{C}_{1},\mathcal{C}_{2})=\mathcal{C}\wedge\mathcal{C}=\mathcal{C}$.
%A similar arguments runs for $\Upsilon$.
%\end{proof}

\begin{prop}
For all $X\in\mathcal{L}_{C}$ $X\subseteq\tau^{-1}_{1}(\tau_{1}(X))$
and for all $Y\in\mathcal{L}_{C_{1}}$, $\tau_{1}(\tau^{-1}_{1}(Y))$.
For all $C\subseteq\mathcal{C}$ we have
$C\subseteq\tau^{-1}_{1}(C_{1})\wedge\tau^{-1}_{2}(C_{2})$
\end{prop}

\begin{proof}
Let $X\in\mathcal{L}_{C}$. Then, if $x\in X$ it follows that
$\tau_{1}(x)\in\tau_{1}(X)$ and so,
$X\subseteq\tau^{-1}_{1}(\tau_{1}(X))$. If $Y\in\mathcal{L}_{C_{1}}$
and $z\in \tau_{1}(\tau^{-1}_{1}(Y))$. Then by definition of
$\tau^{-1}_{1}(Y)$, it follows that $z\in Y$.

Let $C\in\mathcal{L}_{\mathcal{C}}$. Now
$\tau_{1}(C)=C_{1}\in\mathcal{L}_{\mathcal{C}_{1}}$ and
$\tau_{2}(C)=C_{2}\in\mathcal{L}_{\mathcal{C}_{2}}$. Then, it is
apparent that $C\subseteq\tau^{-1}_{1}(C_{1})$ and
$C\subseteq\tau^{-1}_{2}(C_{2})$. And so
$C\subseteq\tau^{-1}_{1}(C_{1})\wedge\tau^{-1}_{2}(C_{2})$.

\end{proof}

\begin{prop}
For all $a,b\in\mathcal{L}_{\mathcal{C}_{1}}$ $\tau^{-1}_{1}(a\wedge
b)=\tau^{-1}_{1}(a)\wedge\tau^{-1}_{1}(b)$, $\tau^{-1}_{1}(a\vee
b)=\tau^{-1}_{1}(a)\vee\tau^{-1}_{1}(b)$. Furthermore,
$\tau^{-1}_{1}$ is an injective function and if
$a,b\in\mathcal{L}_{\mathcal{C}_{1}}$ and $a\subseteq b$, then
$\tau^{-1}_{1}(a)\subseteq\tau^{-1}_{1}(b)$. If $\rho\neq\rho'$ then
$\tau^{-1}_{1}(\rho)\wedge\tau^{-1}_{1}(\rho')=\mathbf{0}$.
\end{prop}

\begin{proof}
Consider the sets $\tau^{-1}_{1}(a\wedge b)$ and
$\tau^{-1}_{1}(a)\wedge\tau^{-1}_{1}(b)$. Then,
$x\in\tau^{-1}_{1}(a)$ and $x\in\tau^{-1}_{1}(b)$. If
$x\in\tau^{-1}_{1}(a\wedge b)$, then $\tau_{1}(x)\in a\wedge
b\subseteq a$, and we obtain also $\tau_{1}(x)\in a\wedge b\subseteq
b$. This means that $x\in\tau^{-1}_{1}(a)$ and
$x\in\tau^{-1}_{1}(b)$. So we have $\tau^{-1}_{1}(a\wedge
b)\subseteq\tau^{-1}_{1}(a)\wedge\tau^{-1}_{1}(b)$. On the other
hand, if $x\in\tau^{-1}_{1}(a)\wedge\tau^{-1}_{1}(b)$, then
$x\in\tau^{-1}_{1}(a)$ and $x\in\tau^{-1}_{1}(b)$. This means that
$\tau_{1}(x)\in a$ and $\tau_{1}(x)\in b$, and so, $\tau_{1}(x)\in
a\wedge b$. This means that $x\in\tau^{-1}_{1}(a\wedge b)$. This
concludes the proof that $\tau^{-1}_{1}(a\wedge
b)=\tau^{-1}_{1}(a)\wedge\tau^{-1}_{1}(b)$.

If $x\in\tau^{-1}_{1}(a)\vee\tau^{-1}_{1}(b)$ then
$x=\alpha\rho+\beta\rho'$, with $\tau_{1}(\rho)\in a$ and
$\tau_{1}(\rho')\in b$. So
$\tau_{1}(x)=\alpha\tau_{1}(\rho)+\beta\tau_{1}(\rho')\in a\vee b$.
This means that $x\in\tau^{-1}_{1}(a\vee b)$, and we have
$\tau^{-1}_{1}(a\vee
b)\supseteq\tau^{-1}_{1}(a)\vee\tau^{-1}_{1}(b)$. Now, let
$x\in\tau^{-1}_{1}(a\vee b)$. Then, $\tau_{1}(x)\in a\vee b$. This
means that $\tau_{1}(x)=\alpha\rho+\beta\rho'$ (convex combination),
with $\rho\in a$ and $\rho'\in b$. There exist
$\sigma\in\tau^{-1}_{1}(a)$ and $\sigma'\in\tau^{-1}_{1}(b)$ such
that $\tau_{1}(\sigma)=\rho$ and $\tau_{1}(\sigma')=\rho'$. Then
$\tau_{1}(x)=\alpha\tau_{1}(\sigma)+\beta\tau_{1}(\sigma')$.
$\tau_{1}()$ is a linear function so, the last equality implies
$\tau_{1}(x-(\alpha\sigma+\beta\sigma'))=0$. Then, there exists
$\varsigma\in\mathbf{Ker}(\tau_{1}())$ such that
$x=\alpha\sigma+\beta\sigma'+\varsigma$. If $\beta=0$, then
$\alpha=1$ (convex combination), and then,
$x=\sigma\in\tau^{-1}_{1}(a)$, and in that case
$x\in\tau^{-1}_{1}(a)\vee\tau^{-1}_{1}(b)$. If $\beta\neq 0$, we can
put $x=\alpha\sigma+\beta(\sigma'+\frac{1}{\beta}\varsigma)$.
$\tau_{1}((\sigma'+\frac{1}{\beta}\varsigma))=\tau_{1}(\sigma')+0\in
b$, and so $\sigma'+\frac{1}{\beta}\varsigma\in\tau^{-1}_{1}(b)$.
This proves that $x\in\tau^{-1}_{1}(a)\vee\tau^{-1}_{1}(b)$, and
thus $\tau^{-1}_{1}(a\vee
b)\subseteq\tau^{-1}_{1}(a)\vee\tau^{-1}_{1}(b)$

Let $a$ and $b$ be two propositions such that $a\neq b$. Suppose
that $\tau^{-1}_{1}(a)=\tau^{-1}_{1}(b)$. If $a\neq b$, there exists
$\rho_{a}\in a$ such that $\rho_{a}\notin b$. It is clear that
$\tau^{-1}_{1}(\rho_{a})\subseteq\tau^{-1}_{1}(a)=\tau^{-1}_{1}(b)$
and then, there exists $\rho\in\tau^{-1}_{1}(b)$ such that
$\tau_{1}(\rho)=\rho_{a}$. But by definition of $\tau^{-1}_{1}(b)$,
we would have that $\rho_{a}\in b$, a contradiction. Thus, we have
$\tau^{-1}_{1}(a)\neq\tau^{-1}_{1}(b)$.

If $a\subseteq b$, suppose that $x\in\tau^{-1}_{1}(a)$. Then
$\tau_{1}(x)\in b$, and so $x\in\tau^{-1}_{1}(b)$ also. If
$x\in\tau^{-1}_{1}(\rho)$, $x\in\tau^{-1}_{1}(\rho')$ and
$\rho\neq\rho'$, then $\rho=\tau_{1}(x)=\rho'$, a contradiction.
\end{proof}

\section{The Difference with Other Approaches and
Conclusions}\label{s:Conclusions}

The problem of compound quantum systems has been widely studied from
different approaches. An important difference of our approach is
that it treats improper mixtures in a different way. In this work,
we made the following reasoning line. In $CM$ the fundamental
description is given by subsets (propositions) of the phase space.
Statistical mixtures \emph{are not fundamental}; they appear as a
limitation in the capability of knowledge of the observer. It is for
that reason that they are expressed as proper mixtures, and so in
the orthodox logical approach they appear in \emph{different
levels}: the propositions lie in the lattice
$\mathcal{L}_{\mathcal{C}\mathcal{M}}$, while the mixtures are
measures over this lattice. Pure states are in one to one
correspondence with the atoms of the lattice, and so they are also
included as elements of the lattice, but mixtures are in a different
level.

The situation turns radically different in $QM$ if we accept that
improper mixtures do not admit an ignorance interpretation. If
$S_{1}$ is in a state represented by an improper mixture $\rho$,
there is no more physical information available for the observer.
$\rho$ represents the actual state for $S_{1}$ and we cannot get
more information, not because of our experimental limitations, but
because of that information \emph{does not exist}. Thus, they should
be represented at the same level as that of pure states, because
they are maximal pieces of information.

On the other hand, influenced by historical reasons, the orthodox
$QL$ approach, still retains the analogy with $CM$ and considers
improper mixtures in a different level than that of the
propositions. Thus, the orthodox $QL$ approach presents some
difficulties when compound quantum systems are involved. We studied
these problems in section \ref{s:Explicamos por que} and we gave a
list of conditions on the structures that we are looking for in
order to solve these difficulties.

The structures presented in this work and in \cite{extendedql},
$\mathcal{L}_{\mathcal{C}}$ and $\mathcal{L}$, do not have these
problems, but on the contrary, they incorporate these quantum
mechanical features explicitly. This is so because they satisfy the
conditions listed in section \ref{s:Explicamos por que}. While each
state (pure or mixed) induces a measure in the lattice of
projections, this has nothing to do with the identification of these
measures with classical mixtures. Indeed, any pure state is not
dispersion free also and so induces a measure over
$\mathcal{L}_{v\mathcal{N}}$. This measure has a radical different
nature from that of classical measures; in our approach, improper
mixtures are in the same status than pure states and induce measures
over $\mathcal{L}_{v\mathcal{N}}$ as well as pure states, but
measures and states are not identified. This situation is very
different from that of $CM$, in which the measures induced by pure
states are trivial.

Our approach -specially $\mathcal{L}_{\mathcal{C}}$- presents itself
as a natural logical and algebraic language for the study of topics
which involve compound quantum systems such as quantum information
processing and decoherence, which concentrate on the study of
$\mathcal{C}$ instead of the lattice of projections. In particular,
we can map states of the compound system into states of its
subsystems at the lattice level, while this cannot be done in the
standard $QL$ approach. Furthermore, $\mathcal{L}_{\mathcal{C}}$ and
$\mathcal{L}$ capture the physics behind the fact that we can mix
states according to the ``mixing principle" of section
\ref{s:introduction to QL}.

As discussed in section \ref{s:interaction in QM}, our construction
shows a new radical difference with classical mechanics, namely,
that of the enlargement of the propositional structure when
interactions are involved, a difference which is not clear in the
standard $QL$ approach.

Moreover, as we showed in section \ref{s:entanglement},
$\mathcal{L}_{\mathcal{C}}$ sheds new light into algebraic
properties of quantum entanglement via the study of the natural
arrows defined between the lattice of the system and its subsystems.
The study of these arrows reveals itself as adequate for the of
algebraic characterization of entanglement.

\vskip1truecm

\noindent {\bf Acknowledgements} \noindent This work was partially
supported by the following grants: PIP N$^o$ 6461/05 (CONICET). We
wish to thank G. Domenech for careful reading and discussions.

\end{document}